\documentclass[12pt,a4paper]{article}

\usepackage[british]{babel}

\usepackage[a4paper,top=2cm,bottom=2cm,left=2.5cm,right=2.5cm,marginparwidth=1.75cm]{geometry}

\usepackage[style=apa, backend=biber,isbn=false,url=false,doi=false,eprint=false]{biblatex} 
\DeclareLanguageMapping{british}{british-apa} 
\DeclareFieldFormat[article]{volume}{\apanum{#1}} 

\usepackage{amsmath}
\usepackage{amsthm}
\usepackage{graphicx}
\usepackage[colorlinks=true, allcolors=blue]{hyperref}
\usepackage{hyperref}
\usepackage[title]{appendix}
\usepackage{mathrsfs}
\usepackage{amsfonts}
\usepackage{booktabs} 
\usepackage{caption}  
\usepackage{threeparttable} 
\usepackage{algorithm}
\usepackage{algorithmicx}
\usepackage{algpseudocode}
\usepackage{listings}
\usepackage{enumitem}
\usepackage{chngcntr}
\usepackage{booktabs}
\usepackage{subcaption}
\usepackage{authblk}
\usepackage[T1]{fontenc}    
\usepackage{csquotes}       
\usepackage{diagbox}
\usepackage{lipsum}         

\usepackage{setspace}
\usepackage{titlesec}
\titleformat{\section} 
  {\normalfont\Large\bfseries}{\thesection.}{1em}{}

\usepackage{float}   
\usepackage{caption} 


\def\P{\mathbb{P}}
\def\E{\mathbb{E}}
\def\k{\kappa_{\alpha}}
\newtheorem{theorem}{Theorem}

\title{Confidence Intervals Based on the Modified Chi-Squared Distribution and its Applications in Medicine}

\author[1]{Mulan Wu}
\author[2]{Mengyu Xu}
\author[3]{Dongyun Kim}
\affil[1]{\small University of Chicago Laboratory High School, Chicago, IL, 60637 USA}
\affil[2]{Department of Statistics and Data Science, College of Sciences, University of Central Florida, Orlando, FL 32816 USA}
\affil[3]{Department of Statistics, College of Engineering and Computing, Fairfax, VA 22030 USA}
\affil[*]{Corresponding author: \texttt{mengyu.xu@ucf.edu}}

\date{}  

\begin{document}
\maketitle

\begin{abstract} 
Small sample sizes in clinical studies arises from factors such as reduced costs, limited subject availability, and the rarity of studied conditions. This creates challenges for accurately calculating confidence intervals (CIs) using the normal distribution approximation. In this paper, we employ a quadratic-form based statistic, from which we derive more accurate confidence intervals, particularly for data with small sample sizes or proportions. Based on the study, we suggest reasonable values of sample sizes and proportions for the application of the quadratic method. Consequently, this method enhances the reliability of statistical inferences. We illustrate this method with real medical data from clinical trials.



\end{abstract}


\section{Introduction}

\indent Confidence intervals (CIs) are useful in clinical studies because they provide a range of values within which the true population parameter is likely to fall, which subsequently emphasizes the importance of ensuring their accuracy and precision. Let $B_1,B_2,\ldots, B_n$ be identical and independent Bernoulli($p$) observations in a random sample, where $p$ is the unknown population proportion of interest. When the sample size is sufficiently large and the proportion is reasonably high ($p > 0.1$), the standard method, based on a limiting normal distribution, can be used to calculate a confidence interval. 

Specifically, let $X=\sum_{i=1}^n B_i$,  the point estimator of $p$ is $\hat p=n^{-1}X$. Also let $q=1-p$ and $\hat q=1-\hat p$. It is known that $n\hat p$ follows a Binomial($n,p$) distribution. It is widely accepted that if $npq$ is large, then
the distribution of $\hat p$ can be approximated by $N(p,n^{-1}pq)$, as a result of the Central Limit Theorem. Therefore, the Wald method constructs the following confidence interval for the proportion $p$ with the nominal confidence level $1-\alpha$:
\begin{align}\label{eq:standard}
    CI_W=\left[\hat p-z_{1-\alpha/2}\sqrt{\frac{1}{n}\hat p \hat q}, \hat p+z_{1-\alpha/2}\sqrt{\frac{1}{n}\hat p\hat q}\right],
\end{align}
where $z_{1-\alpha/2}$ is the $(1-\alpha/2)$-th quantile of the standard Normal distribution. Statistical textbooks suggest various rules of thumb for the application, such as  \begin{align}\label{eq:rule_of_thumb}
    \min(np,nq)\ge 5 \mbox{ or }10,\; \min(n\hat p, nq)\ge 5 \mbox{ or }10,  \qquad \mbox{or } n\hat p\hat q\ge 5 \mbox{ or } 10.
\end{align}

The Wald method is a standard and widely accepted approach when the rule-of-thumb criterion (\ref{eq:rule_of_thumb}) holds. Despite its simplicity, this standard method has several well-recognized limitations. In practice, the standard confidence interval may exhibit an overshoot phenomenon, where the lower limit falls below 0 or the upper limit exceeds 1. Additionally, in cases where $\hat{p} = 0$ or $\hat{p} = 1$, with a substantial probability when $p$ is near $0$ or $1$, the Wald confidence interval degenerates with a width of zero. Moreover, the coverage rate of the standard confidence interval is known to exhibit an oscillation phenomenon (see, e.g., \cite{brown_interval_2001}), which is characterized by sudden and drastic changes due to small changes in $n$ or $p$, when the other parameter remains fixed. The oscillation occurs even for moderately large $n$ and for the values of $p$ that are not near the boundary. When a data set has a small sample size or the proportion is near zero or one, the confidence interval based on the normal approximation can be inaccurate and insufficient. 
 
In this study, we are particularly interested in the data with small sample sizes and near-the-boundary proportions.  Small sample sizes are often seen in clinical studies, as a result of various factors such as limited subject availability, reduced costs, rarity of studied conditions, etc. 

Due to the importance of constructing an accurate confidence interval of the Binomial proportion, a large body of literature has been devoted to the investigation of alternatives to the Wald method. 
Denote $\tilde p=\tilde X/\tilde n, \tilde X=X+z_{1-\alpha/2}^2/2$, and $\tilde n=n+z_{1-\alpha/2}^2$. The Wald confidence interval with a continuity correction is proposed as 
\begin{align*}
     CI_{WCC}=\left[\hat p-z_{1-\alpha/2}\sqrt{\frac{1}{n}\hat p(1-\hat p)+\frac{1}{2n}}, \hat p+z_{1-\alpha/2}\sqrt{\frac{1}{n}\hat p(1-\hat p)+\frac{1}{2n}}\right].
\end{align*}
Wilson's confidence interval, due to \cite{wilson_probable_1927}, is the solution to equating the Wald statistic to the critical value, resulting in Wilson's confidence interval,
\begin{align}\label{eq:Wilson}
    CI_{Wilson}=\tilde p\pm z_{\alpha/2}\frac{\sqrt{n\hat p\hat q+z^2_{1-\alpha/2}/4}}{\tilde n}.
\end{align}
\cite{agresti_approximate_1998} proposed the following confidence interval
\begin{align}\label{eq:AC}
CI_{AC}=\tilde p\pm z_{1-\alpha/2}\sqrt{\frac{\tilde p(1-\tilde p)}{\tilde n}}.
\end{align}
\cite{clopper_use_1934} proposed the Clopper-Pearson (CP) method based on the exact Binomial tests. For a more comprehensive review and comparison, see  \cite{agresti_approximate_1998,agresti_simple_2000,blyth_binomial_nodate,brown_interval_2001,casella_refining_1986,clopper_use_1934,mcgrath_binomial_2024,newcombe_two-sided_1998,pires_interval_2008,reed_better_2007,reiczigel_confidence_2003,tobi_small_2005,vollset_confidence_1993} among others. 

In this paper, we consider the construction of the confidence interval based on the quadratic-form statistic discussed in \cite{xu_pearsons_2019}, which has a $\chi_1^2$ limiting distribution. The proposed confidence interval incorporates a correction to improve the accuracy of the coverage rate.

\section{Methods}\label{sec2}
Now define the quadratic form statistic of the Bernoulli sample as below.
\begin{equation}
\label{eq:xu-eq}
^2\chi := n^{-1}\sum_{1\le i\ne j\le n}\frac{(B_{i}-p)(B_j-p)}{pq}+1.
\end{equation}
In Theorem \ref{thm:chi2}, it is shown that $^2\chi\to_{\mathcal D}\chi^2_1$. Consequently, denoting the $(1-\alpha)th$ percentile of the $\chi^2_1$ as $\k= z_{1-\alpha/2}^2$, we construct a confidence interval, $CI_{new}$, that satisfies 
\[
\P(p\in CI_{new})=\P(^2\chi\le \k).
\]
Routine calculation reveals that 
\begin{align}\label{eq:test_stat}
^2\chi = \frac{1}{p(1-p)}n(\hat p-p)^2-\left(\frac{\hat p}{p}+\frac{1-\hat p}{1-p}-2\right),    
\end{align}
and
\begin{align}\label{eq:CI}
CI_{new}=\frac{(n-1)\hat p+(\k-1)/2}{n+\k-2}\pm \frac{\Big[n\k\hat p\hat q+(\k-1)^2/4-\hat p\hat q\Big]^{1/2}}{n+\k-2}.    
\end{align}

The proposed confidence interval (\ref{eq:CI}) is symmetric, with a correction of the sample proportion in the center. The shift is asymptotically negligible and adapts to the value of the sample proportion. That is, the correction term is positive for small $\hat p$'s and negative for large $\hat p$'s.

In addition, the proposed confidence interval (\ref{eq:CI}) is free from the overshooting issue, i.e., $L_\alpha\ge 0$ and $U_\alpha \le 1$.

The margin of error of $CI_{new}$, i.e., $\Big[n\k\hat p\hat q+(\k-1)^2/4-\hat p\hat q\Big]^{1/2}$,  is always positive even when $\hat p\hat q=0$, except at the rarely used $68.27\%$ confidence level. This property provides a remedy for the degeneracy of the standard confidence interval in extreme cases where $\hat p\hat q=0$. Furthermore, the width of the proposed confidence interval is asymptotically equivalent to that of the standard confidence interval when the sample proportion is away from the boundary for large sample sizes. For small $n$ with $\hat p\hat q$ not near the boundary, the proposed confidence interval can be more parsimonious. 
In Figure \ref{fig:CI_width}, we present a comparison of the margins of error of the proposed and standard $0.95$ confidence intervals as functions of the observed values of $\hat p\hat q$ and the sample size $n$, for $5\le n\le 100$. A discussion of the expected margin of error is provided in Section \ref{sec:cover_rate+ME}. 
\begin{figure}
    \centering
    \includegraphics[width=0.8\linewidth]{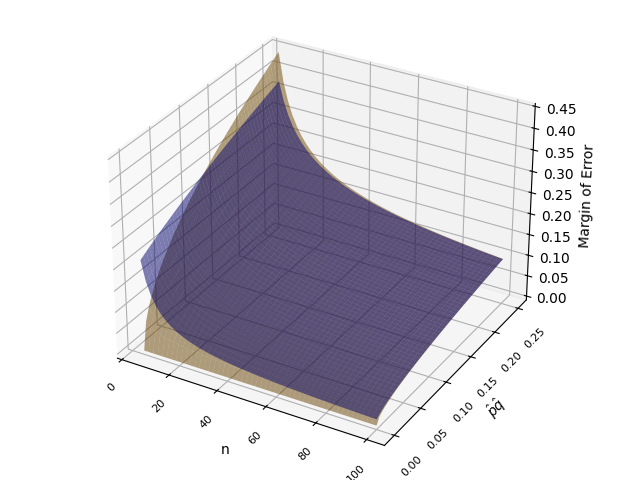}
    \caption{The margin of error of the $95\%$ confidence intervals as functions of $n$ and the value of $\hat p\hat q$. The width of $CI_{W}$ is shown in brown, and that of $CI_{new}$ in blue, respectively. }
    \label{fig:CI_width}
\end{figure}
  

Theorem \ref{thm:chi2} below provides the theoretical guarantee for the asymptotically honest coverage rate of the proposed confidence interval.

\begin{theorem}\label{thm:chi2}
Suppose $np\to \infty$ and $n(1-p)\to\infty$ as $n\to\infty$. Then
\begin{align}
    \sup_t|\P(^2\chi\le t)-F_{\chi_1^2}(t)|\to 0,
\end{align}
where $F_{\chi_1^2}(t)=\P(\chi^2_1\le t)$ is the distribution function of the $\chi_1^2$ random variable. 
\end{theorem}
\begin{proof}[Proof of Theorem \ref{thm:chi2}]
   The theorem is a consequence of Theorem 2 in \cite{xu_pearsons_2019}. 
   Let $X_i=(B_i-p)(p^{-1/2},-q^{-1/2})^\top$. Then
    \[
    ^2\chi=n^{-1}\sum_{i\ne j\le n}X_i^\top X_j+1.
    \]
    If $np\to \infty$ and $n(1-p)\to\infty$ as $n\to\infty$, then 
   \[
   L_1:=\frac{p^{-1}+q^{-1}}{2^{3/2}n}\to 0.
   \]
   Then the result follows from Theorem 2 in \cite{xu_pearsons_2019} applies.
\end{proof}
It follows that
 \begin{align}
        \lim_{n\to\infty} \P(p\in CI_{new})=1-\alpha,
    \end{align}
under the conditions of Theorem \ref{thm:chi2}.


\section{Coverage Rate and Expected Margin of Error}\label{sec:cover_rate+ME}
Since that $X\sim Binomial(n,p)$, we compute the coverage probability and the expected margin of error in this section for the proposed confidence interval and the following popular intervals: standard CI, the Agresti-Coull CI and Wilson's CI. The coverage probability and the expected margin of error for a confidence interval $CI_{\alpha,n}(X)=[L_{\alpha,n}(X),U_{\alpha,n}(X)]$ are respectively computed in (\ref{eq:coverage prob})  and (\ref{eq:ME}).
\begin{eqnarray}\label{eq:coverage prob}
    C(n,p)&:=
    \sum_{k=0}^n \mathbf1 \{p\in[L_{\alpha,n}(k),U_{\alpha,n}(k)]\} p_{n,p}(k).\\ \label{eq:ME}
    \E_{n,p}\Big(ME(CI_{\alpha,n}(X))\Big)&:=\sum_{k=0}^n 2^{-1}(U_{\alpha,n}(k)-L_{\alpha,n}(k))p_{n,p}(k).
\end{eqnarray}
Figure \ref{fig:coverage_prob} and Figure \ref{fig:E_ME} present the coverage probabilities and the expected margins of error for standard CI, the Agresti-Coull CI and Wilson's CI with a 0.95 norminal confidence interval. It can be found that all of the four CIs exhibit an oscillation behavior, while the standard confidence interval is outperformed by all the other CIs in terms of the sufficiency and accuracy of coverage probability. Compared to the Agresti-Coull and Wilson's CI, the proposed method is advantageous for $p$'s near $0$ and $1$, especially for small $n$'s. Indeed, the proposed method has coverage probabilities always near or above the nominal confidence level for $p$ near the boundary, where it is more parsimonious in terms of the expected length. For $p$'s near the center, the coverage probability of the three methods are similar. Wilson's CI and the Agresti-Coull CI are better in respect to the expected length. In the appendix, we present the coverage rates and the expected margins of error of the four 0.95-CIs against at varying $n=5,6,\ldots,100$, with $p=0.01,0.05,0.1$ and $0.2$ respectively in Figures \ref{fig:coerage_in_n} and \ref{fig:E_ME_in_n}. The observations confirm the findings above. We also present results for confidence levels 0.9 and 0.99 in the appendix. At 0.9 confidence level, the coverage rate of the proposed method is more reliable for small $n$ and $p$ near the boundary ($p\le 0.1$ when $n\le 10$), and Agresti-Coull is preferred when moderately large sample is available. For the $0.99$ confidence level, the proposed method is preferrable as it has sufficient coverage and a parsimonous coverage rate for large and small $n$ across $0\le p\le 1$. 

These findings suggest that while the new method in (\ref{eq:CI}) is particularly beneficial for small proportions and small sample size. We provide the following practical recommendations for the proposed method for small-sample 0.95 confidence intervals. For extremely small samples, where $n\le 10$, the method outperforms the others when $p\le 0.2$ or $p\ge 0.8$. For moderately small samples with $n>10$, the proposed method is advantageous when $p\le 0.1$ or $p\ge 0.9$. For the other cases with moderately large sample size and $p$ off the boundary, the proposed confidence interval has similar performance as Wilson's confidence interval and Agresti-Coull's confidence interval, and they are all preferred over the standard confidence interval.
\begin{figure}
    \centering
    \includegraphics[width=1\linewidth]{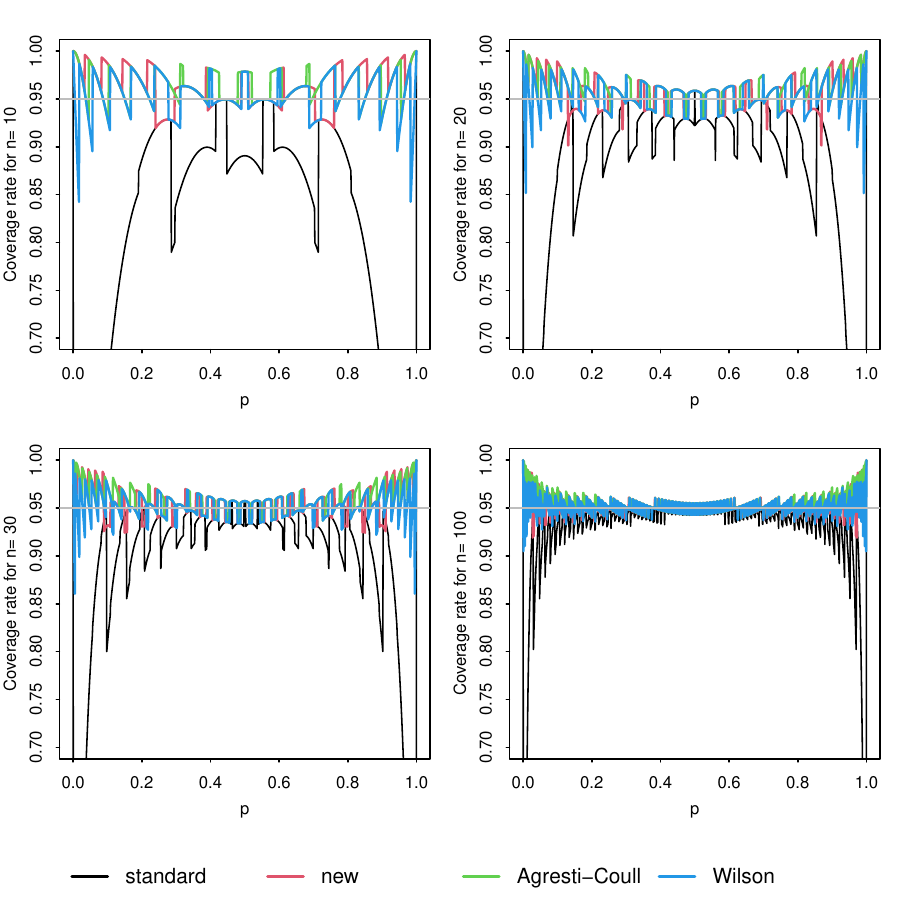}
    \caption{Coverage Probabilities of the 0.95- Standard CI,  proposed CI, the Agresti-Coull CI, and Wilson's CI for $n=10,20,30$ and $100$ at varying $p\in [0,1]$. }
    \label{fig:coverage_prob}
\end{figure}

\begin{figure}
    \centering
    \includegraphics[width=1\linewidth]{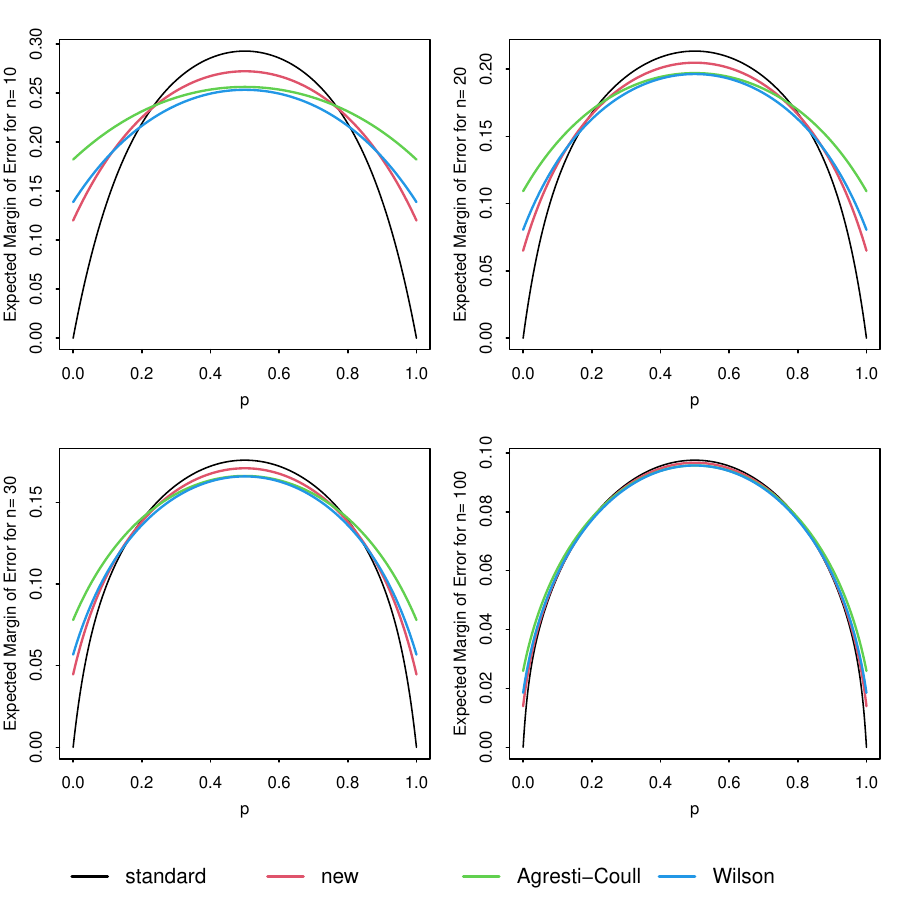}
    \caption{Expected margins of error of the 0.95- Standard CI,  proposed CI, the Agresti-Coull CI, and Wilson's CI for $n=10,20,30$ and $100$ at varying $p\in [0,1]$. }
    \label{fig:E_ME}
\end{figure}
\section{Simulation Results}\label{sec3}
In the simulation study, we generate Binomial($n,p$) random variables for $N=10,000$ times and compute the empirical coverage probability. The performance of the proposed confidence interval is evaluated at $n=10, 30, 50, 100$ and $p=0.01,0.05,0.09$, and is compared with that of the standard confidence interval.

\begin{table}[!ht]
 \caption{Comparison of Empirical Probabilities Across Sample Sizes and Proportions with Different Confidence Levels ($1-\alpha$)}
    \centering
        \begin{tabular*}{\columnwidth}{@{\extracolsep\fill}lllllllll@{\extracolsep\fill}}
    \toprule
    (n,p)&\multicolumn{4}{c}{$1-\alpha=0.95$}&\multicolumn{4}{c}{$1-\alpha=0.99$}\\
\cmidrule(r){2-5} \cmidrule(r){6-9}
& Standard & New & AC & Wilson & Standard & New & AC & Wilson \\ 
        (10,0.05) & 0.402 & 0.988 & 0.988 & 0.912 & 0.389 & 0.989 & 0.989 & 0.989 \\ 
        (10,0.1) & 0.647 & 0.987 & 0.927 & 0.927 & 0.637 & 0.987 & 0.987 & 0.987 \\ 
        (10,0.2) & 0.884 & 0.968 & 0.968 & 0.968 & 0.859 & 0.994 & 0.994 & 0.994 \\ 
        (30,0.05) & 0.783 & 0.985 & 0.985 & 0.939 & 0.783 & 0.997 & 0.997 & 0.997 \\ 
        (30,0.1) & 0.81 & 0.933 & 0.975 & 0.975 & 0.957 & 0.993 & 0.998 & 0.998 \\ 
        (30,0.2) & 0.946 & 0.963 & 0.963 & 0.963 & 0.989 & 0.997 & 0.997 & 0.997 \\ 
        (50,0.05) & 0.92 & 0.988 & 0.962 & 0.962 & 0.924 & 0.997 & 0.997 & 0.997 \\ 
        (50,0.1) & 0.877 & 0.941 & 0.971 & 0.971 & 0.965 & 0.997 & 0.997 & 0.997 \\ 
        (50,0.2) & 0.939 & 0.952 & 0.952 & 0.952 & 0.98 & 0.998 & 0.998 & 0.998 \\ 
        (100,0.05) & 0.878 & 0.934 & 0.965 & 0.965 & 0.963 & 0.996 & 0.999 & 0.999 \\ 
        (100,0.1) & 0.933 & 0.956 & 0.972 & 0.938 & 0.977 & 0.998 & 0.998 & 0.998 \\ 
        (100,0.2) & 0.933 & 0.954 & 0.94 & 0.94 & 0.994 & 0.997 & 0.998 & 0.998 \\ 
        \bottomrule
    \end{tabular*}
    
 \label{tab1}
 \begin{tablenotes}
 \item A summary of the data generated using the methods discussed above.
 \item[1] Entries are rounded to three significant figures.
 \end{tablenotes}
\end{table}


\begin{figure}[!ht]
    \centering
    \includegraphics[width=0.9\linewidth]{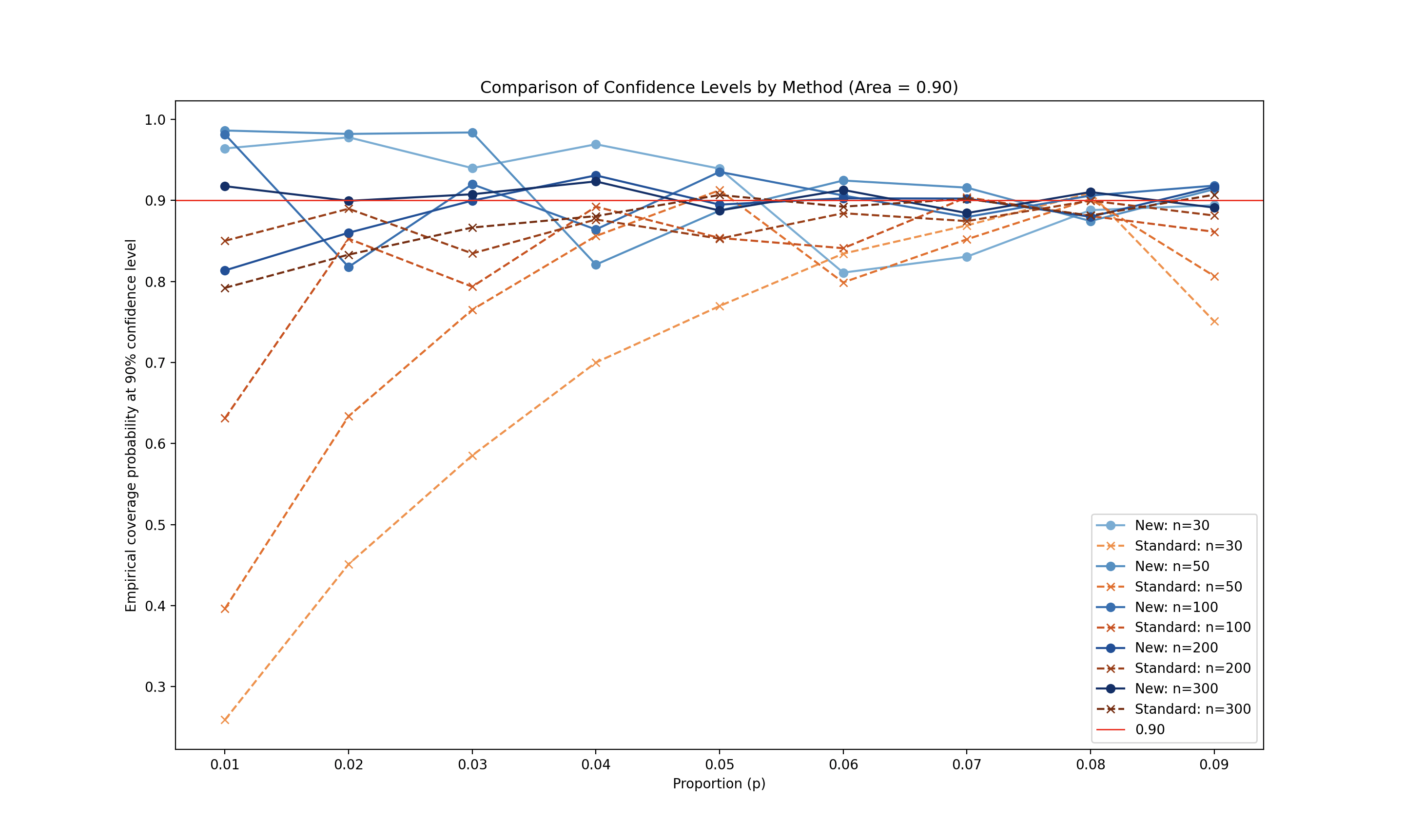}
    \caption{The full data set is graphed. The blue-shaded points represent the empirical probabilities calculated from bounds found using the correction term method, while the orange-shaded points represent empirical probabilities calculated from the normal distribution method. Confidence level = 0.90.}
    \label{fig:fig_0.9}
\end{figure}
\newpage
\begin{figure}[!ht]
    \centering
    \includegraphics[width=0.9\linewidth]{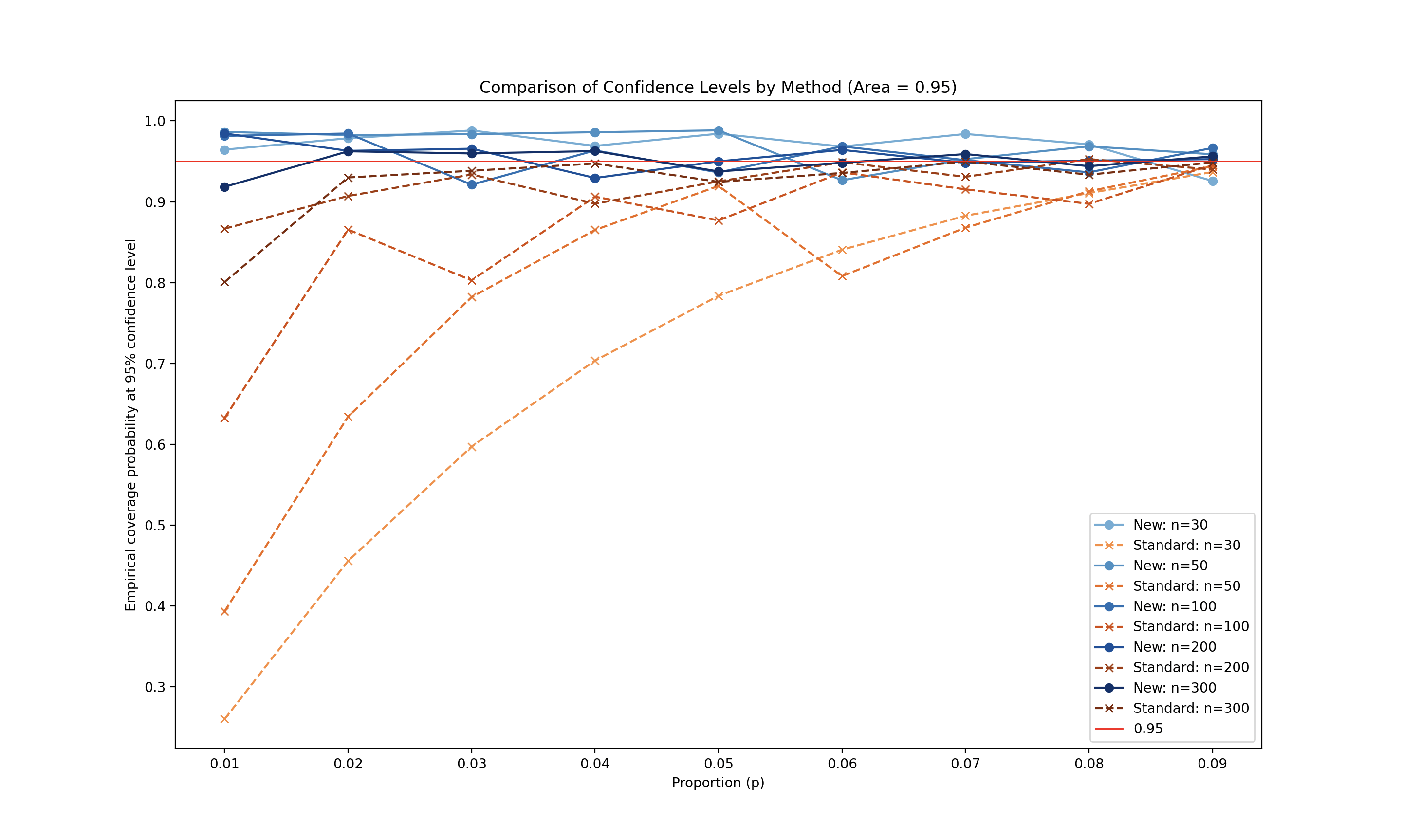}
    \caption{The full data set is graphed. The blue-shaded points represent the empirical probabilities calculated from bounds found using the correction term method, while the orange-shaded points represent empirical probabilities calculated from the normal distribution method. Confidence level = 0.95.}
    \label{fig:fig_0.95}
\end{figure}

\begin{figure}[!ht]
    \centering
    \includegraphics[width=0.9\linewidth]{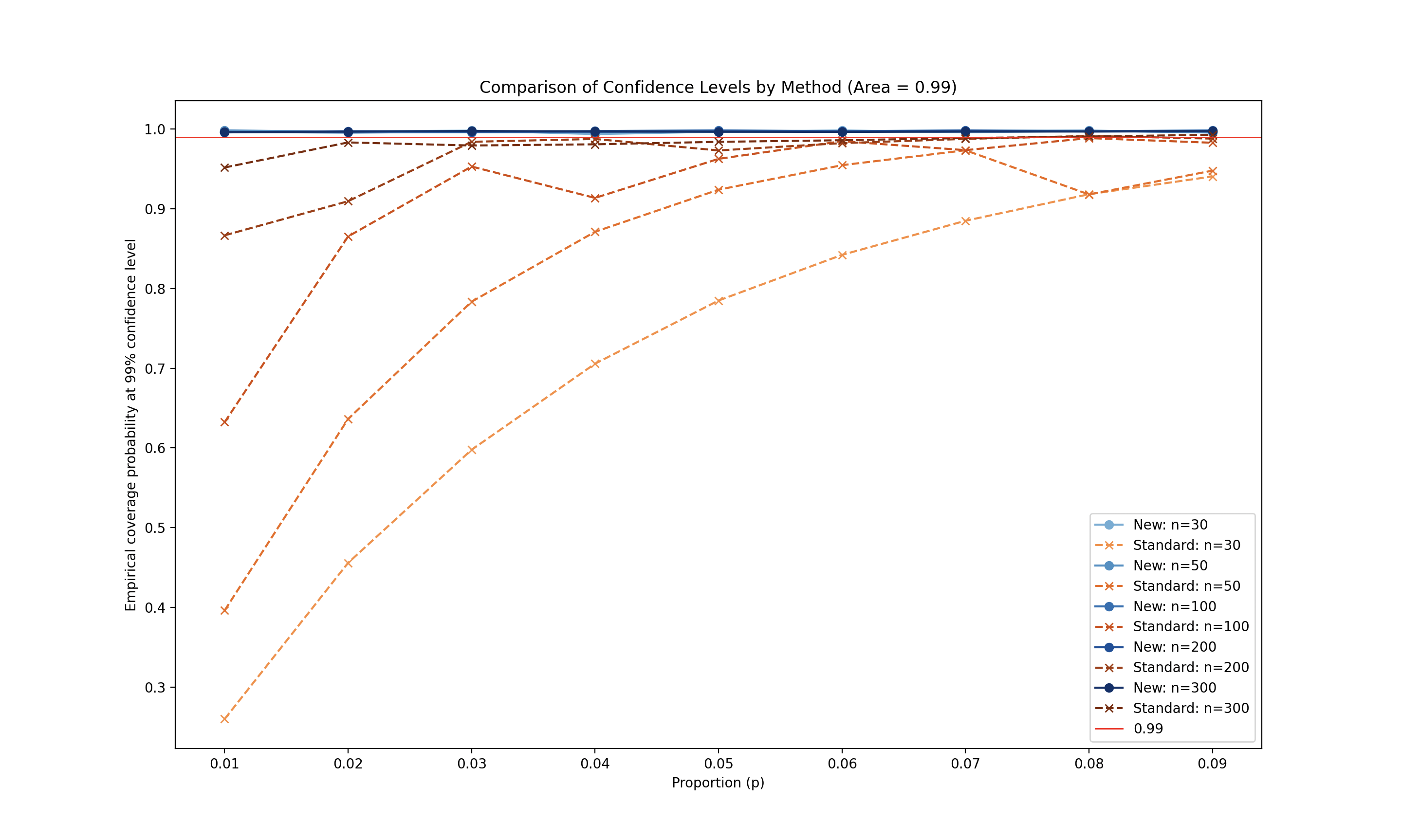}
    \caption{The full data set is graphed. The blue-shaded points represent the empirical probabilities calculated from bounds found using the correction term method, while the orange-shaded points represent empirical probabilities calculated from the normal distribution method. Confidence level = 0.99.}
    \label{fig:fig_0.99}
\end{figure}

\newpage

When examining cases where the proportion is 0.01, it becomes evident that the standard method lacks consistency across different sample sizes. In particular, the standard method performs poorly when both the sample size and the proportion are small. This issue is most apparent in the case where the sample size is $30$ and the proportion is $0.01$ as seen in Table \ref{tab1} where the standard method’s confidence interval deviates significantly from expected values. In contrast, the new method provides more stable and reliable estimates in such scenarios. However, as the proportion increases, the discrepancy between the two methods diminishes. For instance, in the case where the proportion is larger, as seen in the final row of Table \ref{tab1}, the confidence intervals generated by both methods are nearly identical. These findings suggest that while the new method, which incorporates a correction term, is particularly beneficial for small sample and small proportions.

\section{Real Data Application}
To demonstrate the application of our method, we used data from the HF-ACTION clinical trial, a multicenter, two-arm, unblinded randomized controlled trial (RCT) reported in \cite{oconnor_efficacy_2009}. This study examined whether exercise training could reduce all-cause mortality or hospitalization in patients with left ventricular dysfunction and heart failure. The primary endpoint of the trial was the time to death or all-cause hospitalization. 

We first analyzed the overall dataset to compare the performance of the new method against the standard approach. The total death proportion in the HF-ACTION dataset was approximately 0.1692. Since this proportion is relatively large, the confidence intervals produced by both methods were quite similar. Using the standard method, the confidence interval bounds were $(0.1530, 0.1853)$, while the new method produced slightly adjusted bounds of $(0.1535, 0.1857)$. This similarity suggests that for larger proportions, the new method aligns closely with the traditional approach.

To further illustrate the benefits of the new method, particularly in cases with smaller sample sizes and lower event proportions, we conducted a subgroup analysis using female participants from Region 3. In this subset, the control group had a death proportion of 0.111, while the treatment group had a much lower death proportion of 0.0333. Given the small sample size and the low event rate in the treatment group, the new method provided more stable and reliable confidence intervals compared to the standard method. 

For the control group, the standard method produced confidence interval bounds of $(0.0524, 0.1698)$, whereas the new method provided slightly wider and more robust bounds of $(0.0619, 0.1801)$. In the treatment group, the standard method produced bounds of $(-0.0013, 0.0680)$, which included an infeasible negative value, whereas the new method yielded more reasonable bounds of $(0.009, 0.0826)$. The improved reliability of the new method in this low-proportion setting is evident, as it avoids unrealistic negative bounds while maintaining appropriate coverage. 

This example illustrates that while the new method performs similarly to the standard approach for larger proportions, it provides a clear advantage when analyzing smaller subgroups or rare events. This makes it particularly useful in clinical trial analyses where subgroup assessments are necessary, but the sample sizes are limited.

\section{Conclusions}
In this paper, we propose a new method to construct confidence intervals for Binomial proportions.
The proposed confidence interval is based on a quadratic form statistic with an asymptotic chi-squared distribution. 
As an alternative to the Wald confidence interval, the proposed method mitigates practical issues such as overshooting and degeneration while providing more accurate coverage rates, particularly for small sample sizes or proportions near the boundary.

Our study suggests that while the new method, which incorporates a correction term, is always recommended against the standard confidence interval. Compared to the modified confidence intervals, such as the Wilson confidence interval and the Agresti-Coull confidence interval, it is particularly beneficial for small proportions or proportional near zero or one. In practice, we provide the following suggestion using the proposed method to construct 0.95-confidence interval. For extremely small samples where $n\le 10$, the method is recommended over the others when $p\le 0.2$ or $p\ge 0.8$. For moderately small samples with $n>10$, the proposed method is advantageous when $p\le 0.1$ or $p\ge 0.9$. For the other cases with moderately large sample size and $p$ off the boundary, the proposed confidence interval performs similarly to Wilson's confidence interval and Agresti-Coull's confidence interval, and they are all preferred over the standard confidence interval.

\newpage






\printbibliography

\section*{Appendix}
In Figures \ref{fig:coerage_in_n} and \ref{fig:E_ME_in_n}, we present the coverate rates and expected margin of error for the standard, the proposed, the Agresti-Coull, and the Wilson confidence interval respectively at 0.95 confidence levels against the sample size $n$ at $p=0.01,0.05,0.1$ and $0.2$. The figures confirm that the proposed method has accurate and sufficient coverage rates and parsimonious widths at small $n$ and $p$ near-the-boundary.

We also present results at confidence levels 0.9 and 0.99 in Figures \ref{fig:coverage_prob0.9} to \ref{fig:E_ME_in_n99}. At 0.9 confidence level, the coverage rate of the proposed method is more reliable for small $n$ and $p$ near the boundary ($p\le 0.1$ when $n\le 10$), and Agresti-Coull is preferred when moderately large sample is available. For the $0.99$ confidence level, the proposed method is preferrable as it has sufficient coverage and a parsimonous coverage rate for large and small $n$ across $0\le p\le 1$. 

\begin{figure}
    \centering
    \includegraphics[width=1\linewidth]{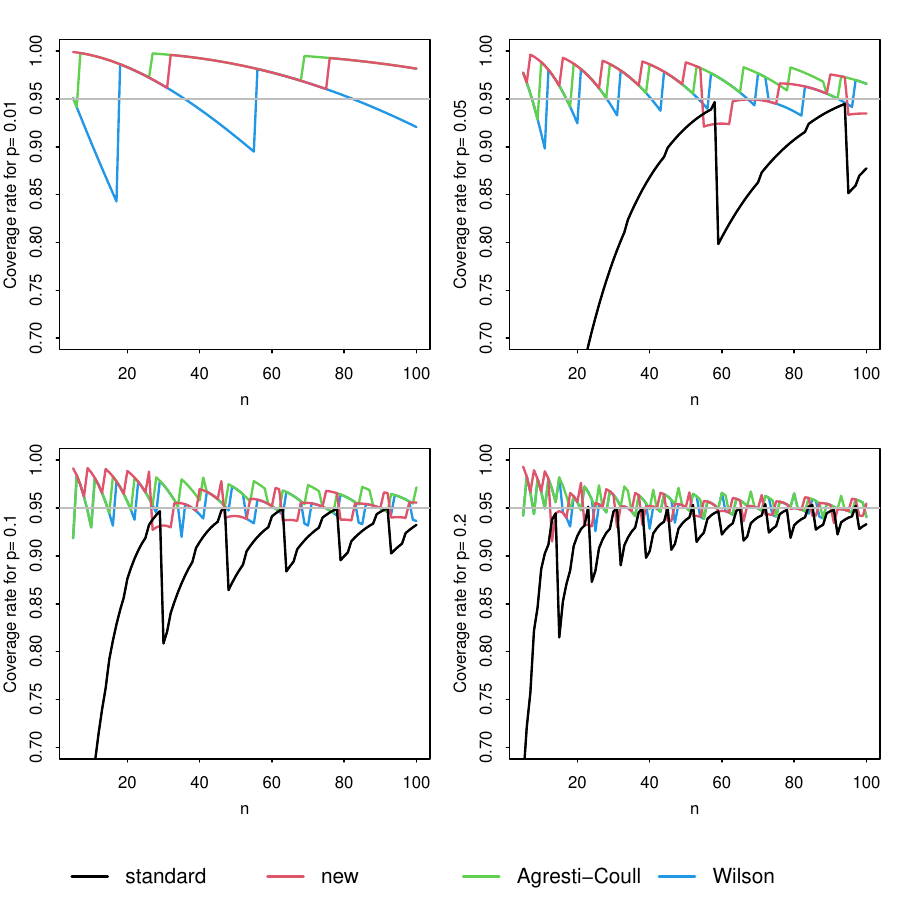}
    \caption{Expected margins of error of the 0.95- Standard CI,  proposed CI, the Agresti-Coull CI, and Wilson's CI for $p=0.01,0.05,0.1$ and $0.2$ at varying $n=5,6,\ldots,100$. }
    \label{fig:coerage_in_n}
\end{figure}
\begin{figure}
    \centering
    \includegraphics[width=1\linewidth]{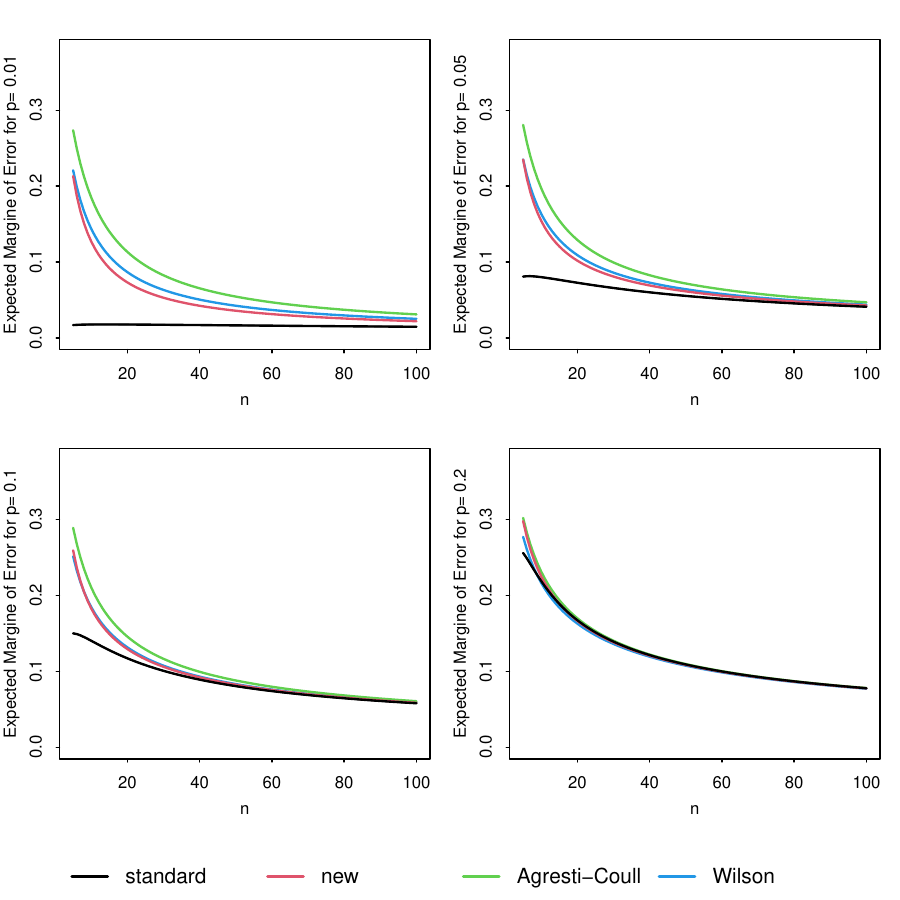}
    \caption{Expected margins of error of the 0.95- Standard CI,  proposed CI, the Agresti-Coull CI, and Wilson's CI for $p=0.01,0.05,0.1$ and $0.2$ at varying $n=5,6,\ldots,100$. }
    \label{fig:E_ME_in_n}
\end{figure}

\begin{figure}
    \centering
    \includegraphics[width=1\linewidth]{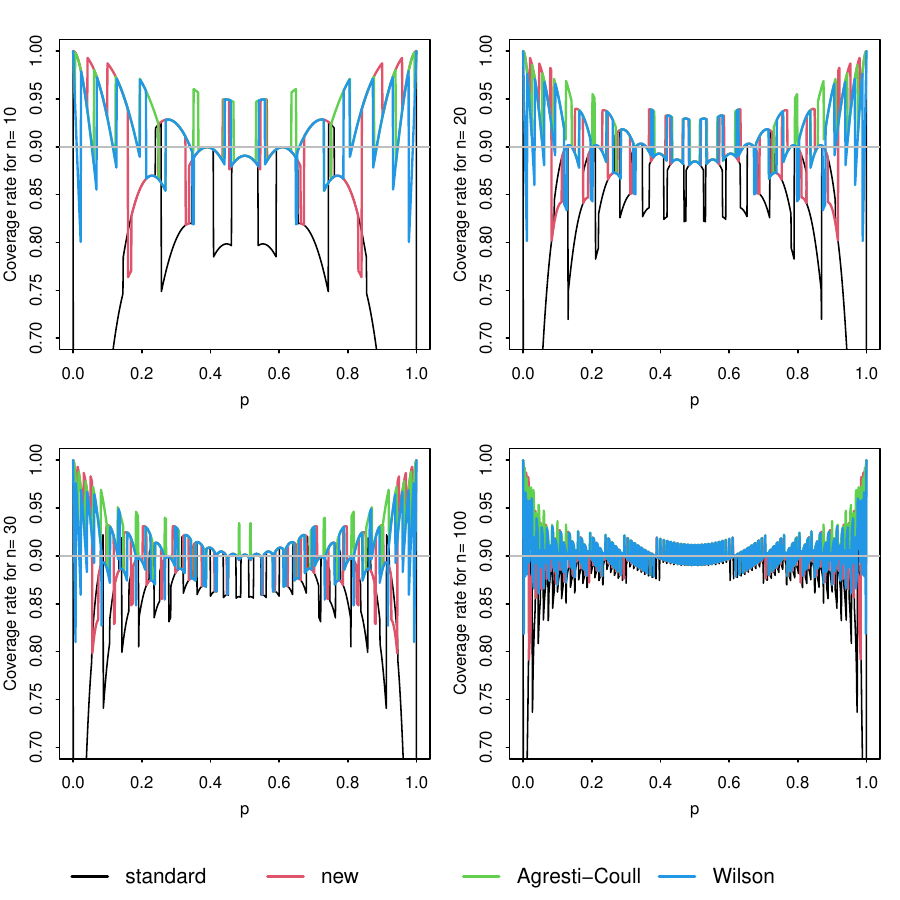}
    \caption{Coverage Probabilities of the 0.9- Standard CI,  proposed CI, the Agresti-Coull CI, and Wilson's CI for $n=10,20,30$ and $100$ at varying $p\in [0,1]$. }
    \label{fig:coverage_prob0.9}
\end{figure}

\begin{figure}
    \centering
    \includegraphics[width=1\linewidth]{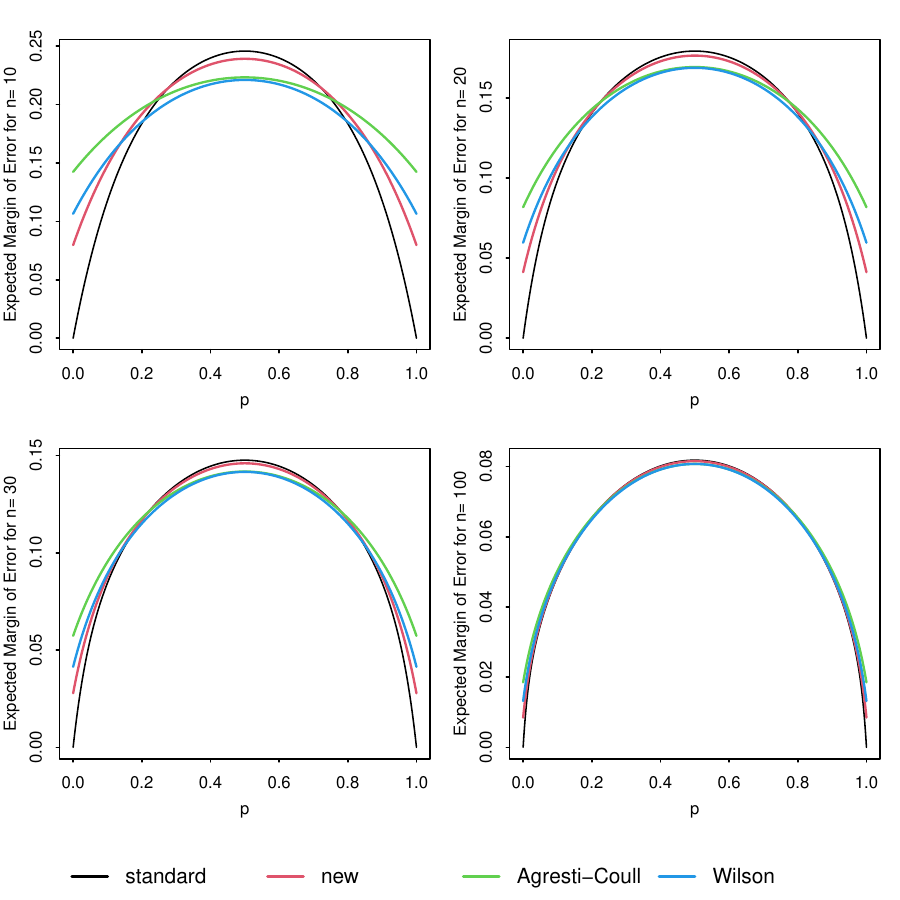}
    \caption{Expected margins of error of the 0.9- Standard CI,  proposed CI, the Agresti-Coull CI, and Wilson's CI for $n=10,20,30$ and $100$ at varying $p\in [0,1]$. }
    \label{fig:E_ME0.9}
\end{figure}
\begin{figure}
    \centering
    \includegraphics[width=1\linewidth]{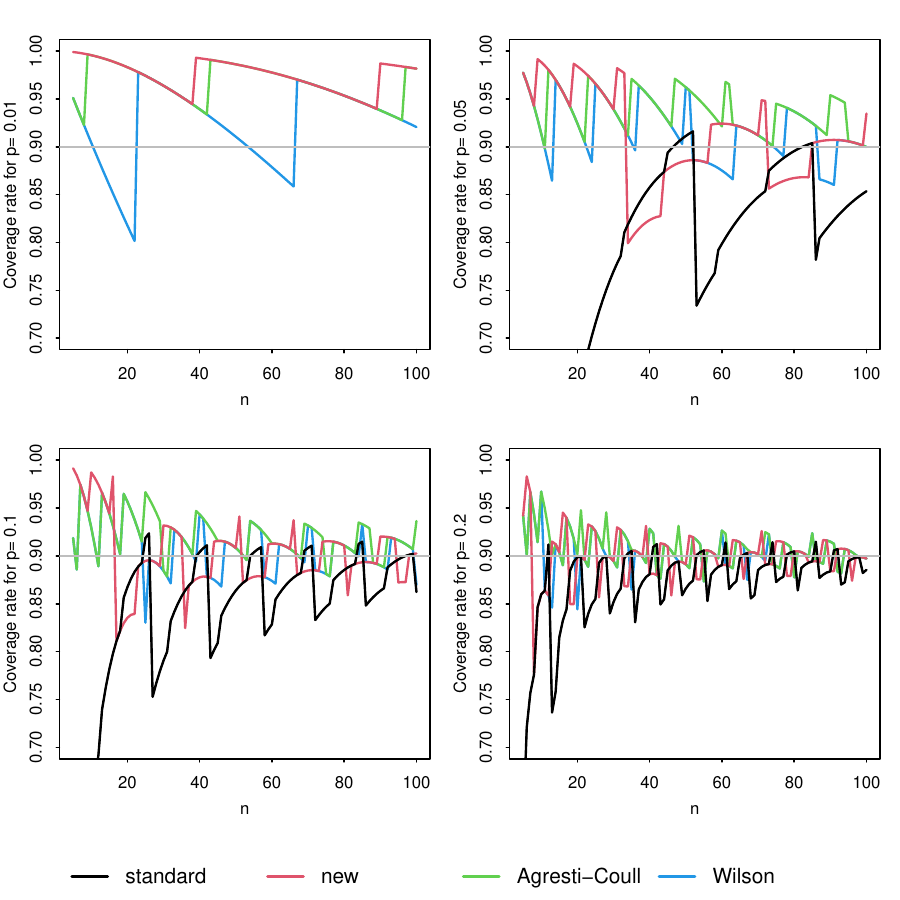}
    \caption{Expected margins of error of the 0.9- Standard CI,  proposed CI, the Agresti-Coull CI, and Wilson's CI for $p=0.01,0.05,0.1$ and $0.2$ at varying $n=5,6,\ldots,100$. }
    \label{fig:coerage_in_n9}
\end{figure}
\begin{figure}
    \centering
    \includegraphics[width=1\linewidth]{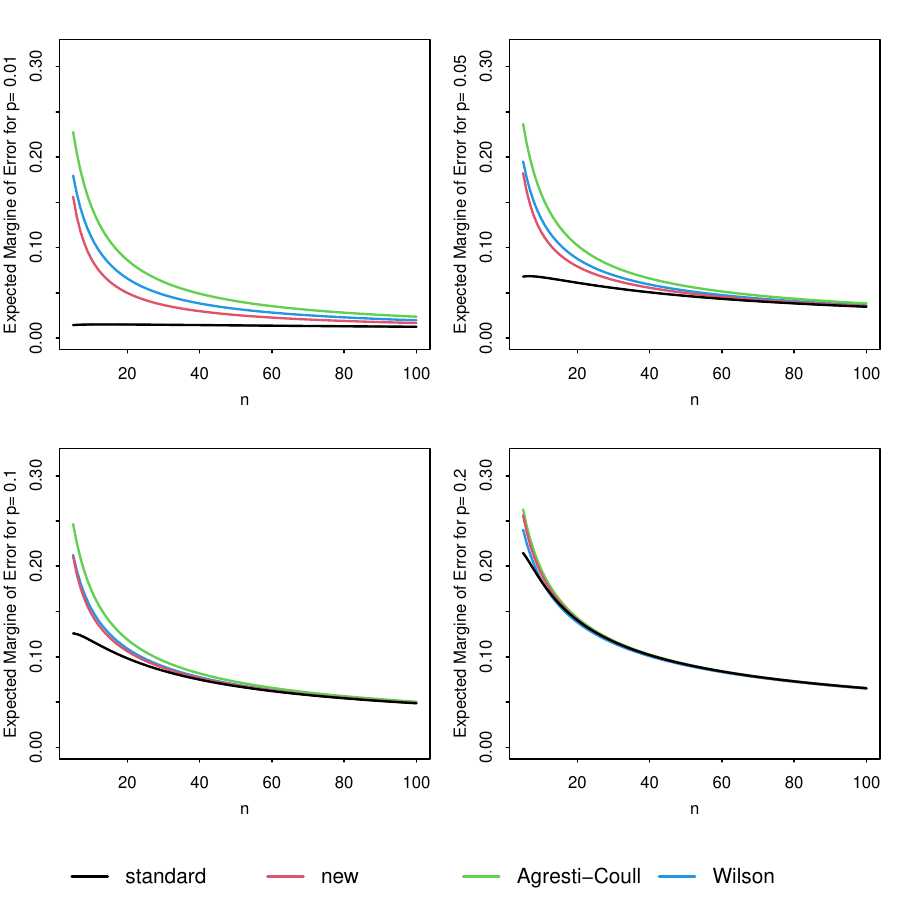}
    \caption{Expected margins of error of the 0.9- Standard CI,  proposed CI, the Agresti-Coull CI, and Wilson's CI for $p=0.01,0.05,0.1$ and $0.2$ at varying $n=5,6,\ldots,100$. }
    \label{fig:E_ME_in_n9}
\end{figure}

\begin{figure}
    \centering
    \includegraphics[width=1\linewidth]{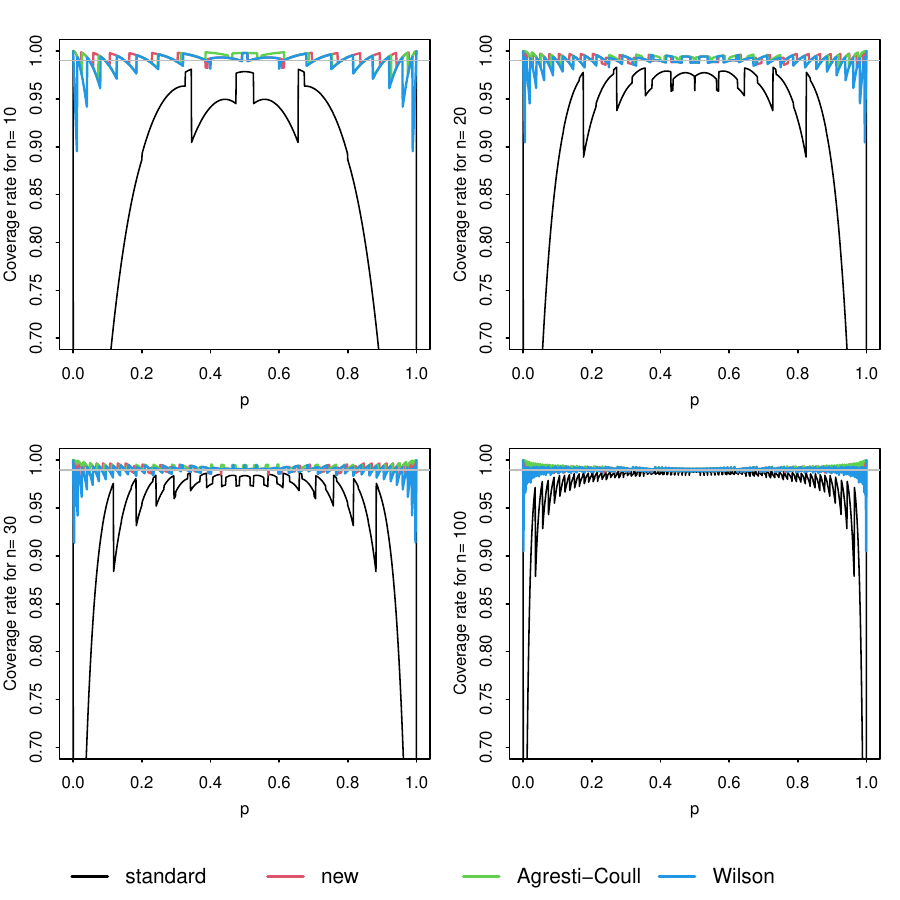}
    \caption{Coverage Probabilities of the 0.99- Standard CI,  proposed CI, the Agresti-Coull CI, and Wilson's CI for $n=10,20,30$ and $100$ at varying $p\in [0,1]$. }
    \label{fig:coverage_prob0.99}
\end{figure}

\begin{figure}
    \centering
    \includegraphics[width=1\linewidth]{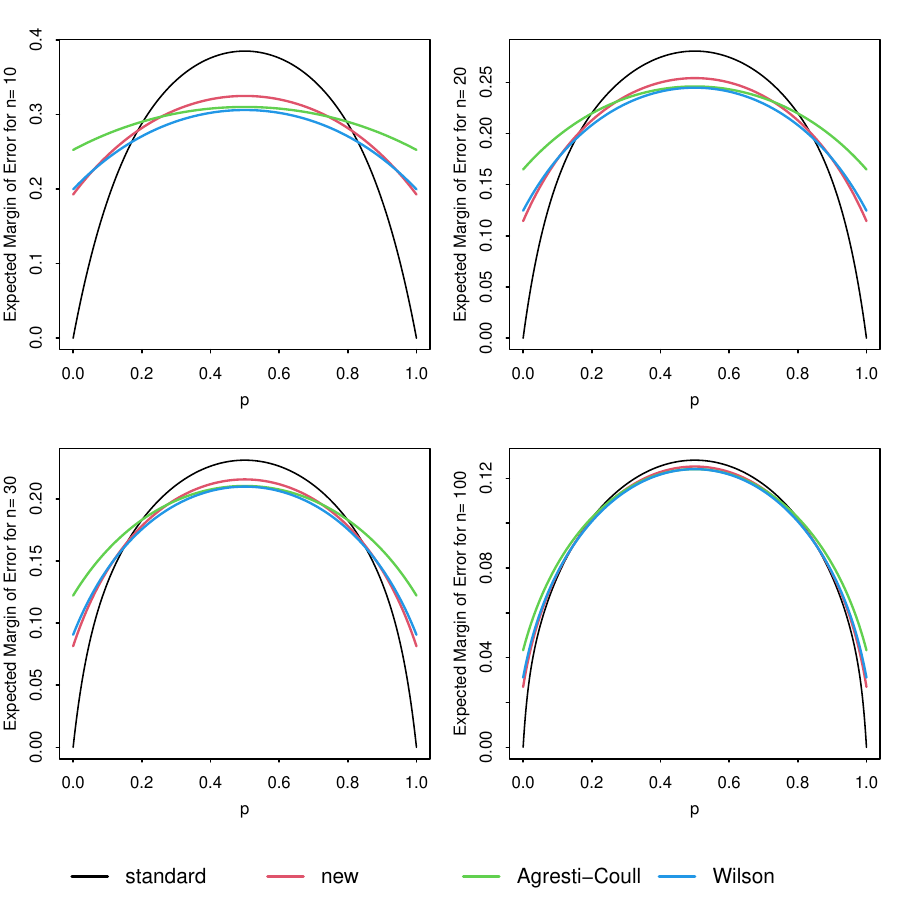}
    \caption{Expected margins of error of the 0.99- Standard CI,  proposed CI, the Agresti-Coull CI, and Wilson's CI for $n=10,20,30$ and $100$ at varying $p\in [0,1]$. }
    \label{fig:E_ME0.99}
\end{figure}
\begin{figure}
    \centering
    \includegraphics[width=1\linewidth]{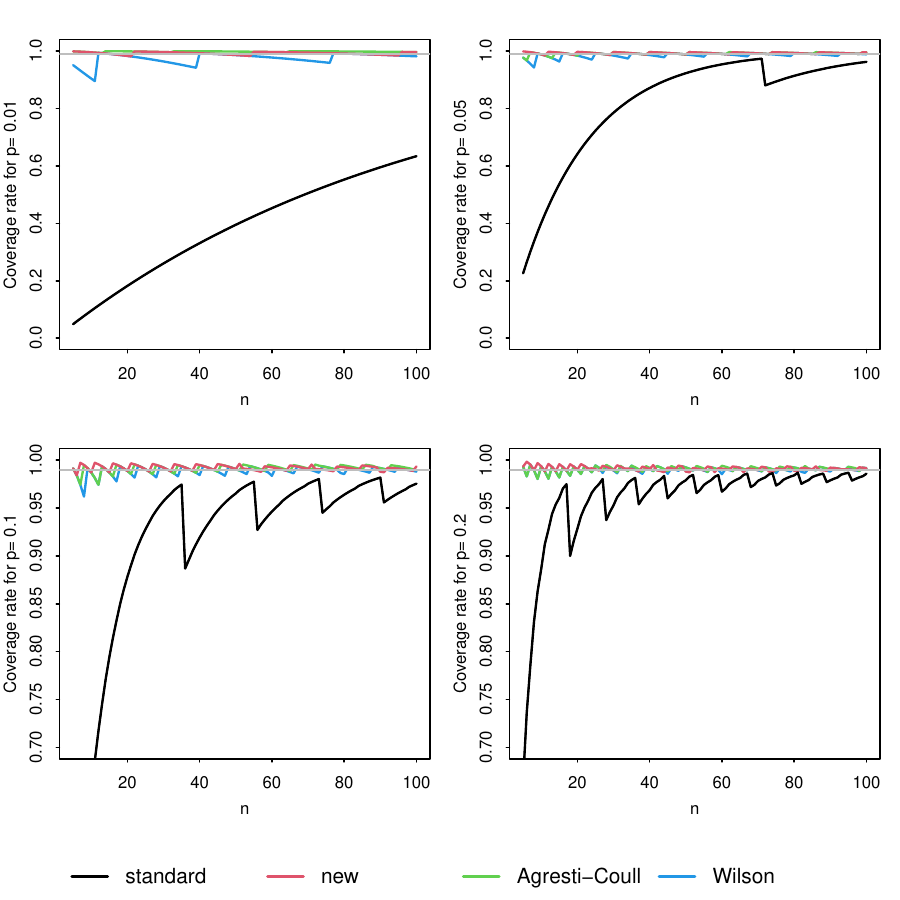}
    \caption{Expected margins of error of the 0.99- Standard CI,  proposed CI, the Agresti-Coull CI, and Wilson's CI for $p=0.01,0.05,0.1$ and $0.2$ at varying $n=5,6,\ldots,100$. }
    \label{fig:coerage_in_n99}
\end{figure}
\begin{figure}
    \centering
    \includegraphics[width=1\linewidth]{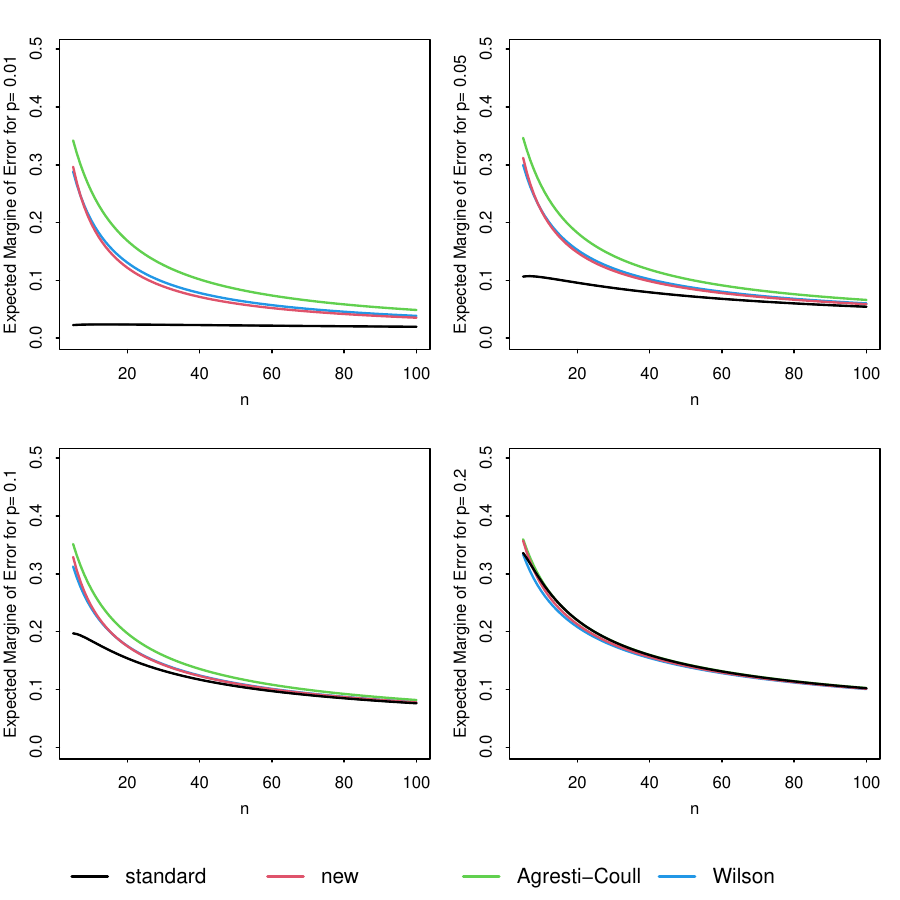}
    \caption{Expected margins of error of the 0.99- Standard CI,  proposed CI, the Agresti-Coull CI, and Wilson's CI for $p=0.01,0.05,0.1$ and $0.2$ at varying $n=5,6,\ldots,100$. }
    \label{fig:E_ME_in_n99}
\end{figure}

\end{document}